\documentclass[runningheads]{llncs}

\usepackage{graphicx}

\usepackage{algpseudocode}
\usepackage[ruled]{algorithm}

\usepackage{environ} 

\usepackage{subcaption} 

\usepackage{amssymb}   
\usepackage{amsmath}   

\spnewtheorem{numclaim}{Claim}{\itshape}{\rmfamily}

\usepackage{hyperref}  

\usepackage{tikz}
\usepackage{xspace}

\newcommand{\chair}[1]{\(Ch(#1)\)}

\newcommand{\yes}{\textsc{Yes}\xspace}
\newcommand{\no}{\textsc{No}\xspace}
\newcommand{\yesinstance}{\yes-instance\xspace}

\newcommand{\classFPT}{\textsf{FPT}\xspace}
\newcommand{\classP}{\textsf{P}\xspace}
\newcommand{\classNP}{\textsf{NP}\xspace}
\newcommand{\NPHard}{\textsf{NP-hard}\xspace}

\newcommand{\Oh}{\ensuremath{\mathcal{O}}}
\newcommand{\OhStar}{\ensuremath{\mathcal{O}^{\star}}}

\newcommand{\ProblemFormat}[1]{\textsc{#1}}
\newcommand{\ProblemIndex}[1]{\index{problem!\ProblemFormat{#1}}}
\newcommand{\ProblemName}[1]{\ProblemFormat{#1}\ProblemIndex{#1}{}\xspace}

\newcommand{\probMR}{\ProblemName{Monopolar Recognition}}
\newcommand{\probLMP}{\ProblemName{List Monopolar Partition}}
\newcommand{\probCFVD}{\ProblemName{Claw-Free Vertex Deletion}}
\newcommand{\probCFED}{\ProblemName{Claw-Free Edge Deletion}}
\newcommand{\probME}{\ProblemName{Monopolar Extension}}
\newcommand{\probTwoSAT}{\ProblemName{\ensuremath{\mathit{2}}-SAT}}

\let\oldproof\proof
\let\oldendproof\endproof
\renewenvironment{proof}
{\oldproof}
{\hfill\qed\oldendproof}

\newcommand{\defdecproblem}[3]{
  \vspace{1mm}
  \noindent\fbox{
    \begin{minipage}{0.96\textwidth}
      \begin{tabular*}{\textwidth}{@{\extracolsep{\fill}}lr} #1 \\ \end{tabular*}
      {\bf{Input:}} #2  \\
      {\bf{Question:}} #3
    \end{minipage}
  }\vspace{1mm}}

\newcommand{\defparproblem}[4]{
  \vspace{1mm}
  \noindent\fbox{
    \begin{minipage}{0.96\textwidth}
      \begin{tabular*}{\textwidth}{@{\extracolsep{\fill}}lr} #1  & {\bf{Parameter:}} #3 \\ \end{tabular*}
      {\bf{Input:}} #2  \\
      {\bf{Question:}} #4
    \end{minipage}
  }\vspace{1mm}}

\newcommand{\repeattheorem}[1]{%
  \begingroup
  \renewcommand{\thetheorem}{\ref{#1}}%
  \expandafter\expandafter\expandafter\theorem
  \csname reptheorem@#1\endcsname
  \endtheorem
  \endgroup
}

\NewEnviron{reptheorem}[1]{%
  \global\expandafter\xdef\csname reptheorem@#1\endcsname{%
    \unexpanded\expandafter{\BODY}%
  }%
  \expandafter\theorem\BODY\unskip\label{#1}\endtheorem
}

\newcommand{\repeatlemma}[1]{%
  \begingroup
  \renewcommand{\thelemma}{\ref{#1}}%
  \expandafter\expandafter\expandafter\lemma
  \csname replemma@#1\endcsname
  \endlemma
  \endgroup
}

\NewEnviron{replemma}[1]{%
  \global\expandafter\xdef\csname replemma@#1\endcsname{%
    \unexpanded\expandafter{\BODY}%
  }%
  \expandafter\lemma\BODY\unskip\label{#1}\endlemma
}

\algrenewcommand\algorithmicrequire{\textbf{Input:}}
\algrenewcommand\algorithmicensure{\textbf{Output:}}

\begin{document}

\title{Faster Algorithms for Graph Monopolarity \thanks{An extended abstract of this work has been accepted at WG 2025 (51st International Workshop on Graph-Theoretic Concepts in Computer Science).}}


\titlerunning{Faster Algorithms for Graph Monopolarity}

\author{Geevarghese Philip\inst{1}\orcidID{0000-0003-0717-7303} \and Shrinidhi Teganahally Sridhara\inst{2}\orcidID{0000-0002-0288-8598}}

\authorrunning{G. Philip and S. T. Sridhara}

\institute{Chennai Mathematical Institute, India and UMI ReLaX \\ \email{gphilip@cmi.ac.in} \\
\url{https://www.cmi.ac.in/~gphilip} \and Univ. Bordeaux, CNRS, Bordeaux INP, LaBRI, UMR 5800, [F-33400 Talence], France
\email{shrinidhi.teganahally-sridhara@u-bordeaux.fr} 
}

\maketitle

\begin{abstract}
  A graph $G = (V,E)$ is \emph{monopolar} if its vertex set admits a partition
  $V = (C \uplus{} I)$ where \(G[C]\) is a \emph{cluster graph} and $I$ is an
  \emph{independent set} in $G$; this is a \emph{monopolar partition} of
  \(G\). 
  The \probMR problem---deciding whether an input graph is monopolar---is known
  to be \NPHard in very restricted graph classes such as sub-cubic planar
  graphs.


  We derive a polynomial-time algorithm that takes (i) a graph \(G=(V,E)\) and
  (ii) a vertex modulator \(S\) of \(G\) to chair-free graphs as inputs, and
  checks whether \(G\) has a monopolar partition \(V=(C\uplus{}I)\) where set
  \(S\) is contained in the cluster part. We build on this algorithm to develop
  fast exact exponential-time and parameterized algorithms for \probMR.

  Our exact algorithm solves \probMR in \(\OhStar(1.3734^{n})\) time on input
  graphs with \(n\) vertices, where the \(\OhStar()\) notation hides polynomial
  factors. In fact, we solve the more general problems \probME and \probLMP
  in \(\OhStar(1.3734^{n})\) time. These are the first improvements over the
  trivial \(\OhStar(2^{n})\)-time algorithms for all these problems. It is known
  that---assuming ETH---these problems cannot be solved in \(\OhStar(2^{o(n)})\)
  time.

  Our \classFPT algorithms solve \probMR in \(\OhStar(3.076^{k_{v}})\) and
  \(\OhStar(2.253^{k_{e}})\) time where \(k_{v}\) and \(k_{e}\) are,
  respectively, the sizes of the smallest vertex and edge modulators of the
  input graph to claw-free graphs. These results are a significant addition to
  the small number of \classFPT algorithms currently known for \probMR.
\end{abstract}




\section{Introduction}\label{sec:introduction}
In this work we derive fast exponential-time and fixed-parameter tractable
(\classFPT) algorithms for recognizing monopolar graphs. All our graphs are
finite, undirected, and simple.
Recall that a graph \(H\) is a \emph{cluster graph} if each connected component
of \(H\) is a complete graph. A partition of the vertex set of a graph
\(G=(V,E)\) into two sets \(C\) and \(I\) is a \emph{monopolar partition} of
\(G\) if the subgraph \(G[C]\) induced by set \(C\) is a cluster graph and the
set \(I\) is an independent set in \(G\). A graph \(G\) is \emph{monopolar} if
it has a monopolar partition. The primary focus of this work is on the
algorithmic problem of recognizing monopolar graphs.

\defdecproblem{\probMR}
{An undirected graph \(G=(V,E)\) on \(n\) vertices.}
{
  If \(G\) is monopolar, then output 
  Yes. 
  Otherwise, output No. }

Monopolar graphs are a special case of the more general class of \emph{polar
  graphs} \cite{ekim2008polarity,ekim2008polar}. A graph is \emph{polar} if its
vertex set can be partitioned into two parts, one of which induces a cluster
graph, and the other a \emph{co-cluster} graph---which is the edge-complement of
a cluster graph; equivalently, a complete multipartite graph. 
Polar graphs were introduced as a common generalization of bipartite graphs and
split graphs~\cite{tyshkevich1985decomposition}. Note that monopolar graphs also
generalize both bipartite graphs (\(C\) is an independent set) and split graphs
(\(C\) is a single clique). While both split graphs and bipartite graphs can be
recognized in polynomial time, \probMR is \NPHard in general
graphs~\cite{farrugia2004vertex}, and in various restricted graph classes such
as triangle-free graphs~\cite{churchley2014polarity}, \(3\)-colourable
graphs~\cite{le2011recognizing}, and triangle-free planar graphs of maximum
degree \(3\)~\cite{le2011recognizing}. We propose fast exponential-time and
\classFPT algorithms that solve \probMR in general graphs. Each of our
algorithms also outputs a monopolar partition, if the input graph is monopolar.

\paragraph*{Exact exponential-time algorithms.} Our first result is a fast exact
exponential-time algorithm for \probMR. Note that we can easily solve \probMR in
\(\OhStar(2^{n})\) time (where the \(\OhStar()\) notation hides polynomial
factors) 
by guessing the correct partition of \(V\) into the sets \(C\) and \(I\). We
significantly improve the running time to \(\OhStar(1.3734^{n})\).
\begin{reptheorem}{repthm:exactAlgo}
  \label{thm:exactAlgo} {\normalfont\probMR} can be solved in
  \(\OhStar(1.3734^{n})\) time: There is an algorithm which takes a graph \(G\)
  on \(n\) vertices as input, runs in \(\OhStar(1.3734^{n})\) time, and
  correctly decides if \(G\) is monopolar. If \(G\) is monopolar, then this
  algorithm also outputs one monopolar partition of the vertex set of \(G\).
\end{reptheorem}
As far as we know, this is the first improvement over the trivial exact
algorithm for this problem. We also obtain a significant improvement---over the
simple exact algorithm---for the more general \probME problem introduced by
Le and Nevries~\cite{LE20141}
. Let \(C',I'\) be two subsets of the vertex set of a graph \(G=(V,E)\). A
monopolar partition \(V=C \uplus{} I\) of \(G\) with \(C'\subseteq{}C\) and
\(I'\subseteq{}I\) is called a \emph{monopolar extension} of the pair $(C',I')$.
We say that graph $G$ is \emph{$(C',I')$-monopolar extendable} if it admits a
monopolar partition that is an extension of $(C',I')$.

\defdecproblem{\probME} {An undirected graph \(G=(V,E)\) on \(n\) vertices, and
  two vertex subsets $C' \subseteq V$ and $I' \subseteq V$.} {If there exists a
  monopolar partition which is a monopolar extension of $(C',I')$, then output
  \yes. Otherwise, output \no.}

\probME asks if a given graph, with some subset of
vertices---\(C'\)---preassigned to the cluster side and some
others---\(I'\)---to the independent set side, has a monopolar partition that
respects these assignments. In other words, it checks if this initial assignment
can be extended to partition the entire graph into a cluster graph and an
independent set while respecting the given constraints. Observe that a graph
\(G\) is monopolar if and only if it is $(\emptyset,\emptyset)$-monopolar
extendable. \probME is thus a generalization of \probMR, and our algorithm for
\probMR in \autoref{thm:exactAlgo} is in fact based on an algorithm that solves
the more general \probME in time \(\OhStar(1.3734^{n})\); see
\autoref{thm:exactAlgoLMP} below.

Churchley and Huang~\cite{churchley2012list} introduced the \probLMP problem as
another generalization of \probMR. Given a graph \(G=(V,E)\) and lists
\(\{\emptyset\neq{}L(v)\subseteq\{\hat{C},\hat{I}\}\;;\;v\in{}V\}\), a
\emph{monopolar partition of $G$ that respects the list function $L$} is a
mapping \(f:V\to{}\{\hat{C},\hat{I}\}\) such that (i) $f(v) \in{} L(v)$ holds
for all $v \in V$, (ii) $f^{-1}(\hat{I})$ induces an independent set in \(G\),
and (iii) $f^{-1}(\hat{C})$ induces a cluster graph in \(G\). Informally put,
the list function \(L\) specifies that some vertices belong to an independent
set and some others belong to a cluster graph, and a monopolar partition that
respects the list function \(L\), is one in which the first set of vertices
indeed belongs to the independent set part, and the second set belongs to the
cluster part. The \probLMP problem seeks to determine whether there exists a
monopolar partition of the graph that respects these initial assignments.

\defdecproblem{\probLMP} {An undirected graph \(G=(V,E)\) on \(n\) vertices and
  a list function $L$ defined on the vertices of \(G\) such that
  \(\emptyset\neq{}L(v)\subseteq \{\hat{C},\hat{I}\}\) for all $v \in V$.} {If
  there exists a monopolar partition of \(G\) that respects the list function
  $L$, then output \yes. Otherwise, output \no.}

Clearly, \probLMP generalizes \probMR: When the lists are of the form
\(L(v)=\{\hat{C},\hat{I}\}\) for all \(v\in{}V\), solving \probLMP is equivalent
to solving \probMR. The problems \probME and \probLMP are computationally
equivalent modulo polynomial time (see \autoref{lem:lmp2me}), and our
exponential speed-up in solving \probME directly transfers to such a speed-up
for \probLMP as well:

\begin{reptheorem}{repthm:exactAlgoLMP}\label{thm:exactAlgoLMP}
  Let \(G\) be a graph on \(n\) vertices.
  \begin{enumerate}
  \item There is an algorithm that takes an instance \((G,C',I')\) of
    {\normalfont\probME} as input, runs in \(\OhStar(1.3734^{n})\) time, and
    correctly decides if \(G\) admits a monopolar partition that is an extension
    of $(C',I')$. If such a monopolar extension exists then this algorithm also
    outputs one such monopolar partition of \(G\).
  \item There is an algorithm that takes an instance \((G,L)\) of
    {\normalfont{\probLMP}} as input, runs in \(\OhStar(1.3734^{n})\) time, and
    correctly decides if \(G\) admits a monopolar partition that respects the
    list function \(L\). If \(G\) does admit such a monopolar partition, then
    this algorithm also outputs one such monopolar partition of \(G\).
  \end{enumerate}
\end{reptheorem}
\smallskip

Assuming the Exponential Time Hypothesis, there is no algorithm that can solve
\probMR---and hence, \probME or \probLMP---in \(2^{o(n)}\)
time~\cite[Prop.~9.1]{kanj2018parameterized}. Any improvement over the running
times of \autoref{thm:exactAlgo} or \autoref{thm:exactAlgoLMP} can therefore
only be in the form of a smaller constant base for the exponential term.

\paragraph*{\classFPT algorithms.} We present two \classFPT algorithms for
\probMR parameterized by measures of \emph{distance from triviality}. Churchley
and Huang showed that \probMR can be solved in polynomial time on claw-free
graphs~\cite{churchley2014polarity}. We show that \probMR can be solved in
\classFPT time for the parameter being two natural \emph{deletion distances} to
claw-free graphs.
\begin{reptheorem}{repthm:MRFPTClawFreeVD}\label{thm:MRFPTClawFreeVD}
    There is an algorithm that takes a graph \(G\) on \(n\) vertices as input,
    runs in \(\OhStar(3.076^{k_{v}})\) time, and correctly decides if \(G\) is
    monopolar. If \(G\) is monopolar, then this algorithm also outputs one
    monopolar partition of the vertex set of \(G\). Here \(k_{v}\) is the
    smallest number of vertices that need to be deleted from \(G\) to obtain a
    claw-free graph.
\end{reptheorem}
\begin{reptheorem}{repthm:MRFPTClawFreeED}\label{thm:MRFPTClawFreeED}
    There is an algorithm that takes a graph \(G\) on \(n\) vertices as input,
    runs in \(\OhStar(2.253^{k_{e}})\) time, and correctly decides if \(G\) is
    monopolar. If \(G\) is monopolar, then this algorithm also outputs one
    monopolar partition of the vertex set of \(G\). Here \(k_{e}\) is the
    smallest number of edges that need to be deleted from \(G\) to obtain a
    claw-free graph.
\end{reptheorem}
These two algorithms do \emph{not} require either the numbers \(k_{v},k_{e}\),
or the corresponding vertex or edge sets, to be given as part of the input.
Previous work on the parameterized complexity of \probMR has focused on a
different notion of distance from triviality, namely, the \emph{number of
  cliques} on the ``cluster side'' of the (unknown) monopolar partition. The
problem has an \classFPT algorithm and a polynomial kernel for this parameter;
we summarize these results later, in the subsection on related work. As far as
we know, these are the only two known results on the parameterized complexity of
\probMR. \autoref{thm:MRFPTClawFreeVD} and \autoref{thm:MRFPTClawFreeED} are
thus a significant addition to the set of known \classFPT results for the
problem.

\paragraph*{Our methods.} As we noted above, an \(\OhStar(2^{n})\)-time
algorithm for \probMR follows more or less directly from the definition of the
problem. We can easily improve this running time to
\(\OhStar(3^{\frac{n}{3}})\approx\OhStar(1.4423^{n})\) as follows. Observe that
if a graph \(G=(V,E)\) has a monopolar partition, then it has such a partition
\(V=C\uplus{}I\) where \(I\) is an \emph{inclusion-maximal} independent set of
\(G\). Indeed, let \(V=C\uplus{}I\) be an arbitrary monopolar partition of
\(G\). If there is no vertex \(x\in{}C\) such that \((I\cup\{x\})\) is
independent in \(G\), then \(I\) is already a maximal independent set of \(G\).
Otherwise, find such a vertex \(x\) and set
\(C\gets{}C\setminus\{x\},I\gets{}I\cup\{x\}\). The new pair \((C,I)\) also
forms a monopolar partition of \(G\). Repeating this process yields a monopolar
partition with the desired property. Thus it is enough to look for a monopolar
partition \(V=C\uplus{}I\) where \(I\) is a \emph{maximal} independent set in
\(G\). Since we can enumerate \emph{all} maximal independent sets of \(G\) in
\(\OhStar(3^{\frac{n}{3}})\) time~\cite{moon1965cliques,johnson1988generating},
we can look for such a monopolar partition in \(\OhStar(3^{\frac{n}{3}})\) time.

As far as we can see, \(\OhStar(3^{\frac{n}{3}})\) seems to be the bound on
``easy'' improvements to the running time of exact exponential-time algorithms
for \probMR. To reduce this further to the \(\OhStar(1.3734^{n})\) bound of
\autoref{thm:exactAlgo}, we exploit structural properties of monopolar graphs,
building in particular upon the work of Le and Nevries who showed that \probME
can be solved in polynomial time on \emph{chair-free}
graphs~\cite[Corollary~3]{LE20141}. We significantly strengthen this result of
Le and Nevries; we show that we can solve \probME in polynomial time in
\emph{general graphs} \(G\), if the specified subset \(C'\) is a \emph{vertex
  modulator to chair-free graphs}\footnote{See \autoref{sec:prelims} for
  definitions.}:


 \begin{reptheorem}{repthm:MEClawFreeDS}\label{thm:MEClawFreeDS}
     There is a polynomial-time algorithm that solves {\normalfont\probME} for
     instances \((G,(C',I'))\) where \(C'\) is a vertex modulator of \(G\) to
     chair-free graphs. If graph \(G\) has a monopolar partition that is an
     extension of \((C',I')\), then this algorithm also outputs one such
     partition.
 \end{reptheorem}

 But does such a ``good'' vertex modulator \(C'\) necessarily exist in
 \emph{every} monopolar graph \(G\)? It turns out that it does, and that in fact
 we may assume \(C'\) to be an \emph{inclusion-minimal} vertex modulator to
 chair-free graphs:


 \begin{replemma}{replem:GoodChairFreeDS}\label{lem:GoodChairFreeDS}
     Let \(G=(V,E)\) be a monopolar graph. There exists a subset
     \(C'\subseteq{}V\) of its vertices such that (i) \(C'\) is a vertex
     modulator of \(G\) to chair-free graphs; (ii) no proper subset of \(C'\) is
     a vertex modulator of \(G\) to chair-free graphs, and (iii) graph \(G\) has
     a monopolar partition \(V=(C\uplus{}I)\) where \(C'\subseteq{}C\).
 \end{replemma}

 It is \emph{not} true that \emph{every} vertex modulator to chair-free graphs
 can be made to ``live inside'' the cluster part of some monopolar partition of
 \(G\). Indeed, consider a \emph{chair-free} monopolar graph \(\hat{G}\). It is
 not the case that arbitrary subsets of its vertices---all of which are,
 trivially, modulators to chair-free graphs---belong to the cluster part in some
 monopolar partition of \(\hat{G}\). What \autoref{lem:GoodChairFreeDS} says is
 that there is \emph{at least one such}---inclusion-minimal---modulator in every
 monopolar graph.

 \probMR thus reduces to the problem of finding a minimal vertex modulator to
 chair-free graphs that belongs to the ``cluster part'' of \emph{some} (unknown)
 monopolar partition of the input graph, or ruling out that such a modulator
 exists. One way to do this would be to enumerate all minimal modulators of
 \(G\) to chair-free graphs and check if any of them can be extended to a
 monopolar partition of the form stated in \autoref{lem:GoodChairFreeDS}, using
 \autoref{thm:MEClawFreeDS} while setting \(I'=\emptyset\). A straightforward
 way to enumerate all such minimal vertex modulators is to (i) locate a chair,
 (ii) guess the partition of its vertices into those that belong to the minimal
 modulator and those which don't, and (iii) recurse on the rest of the graph.
 This is a \(31\)-way branching---every non-empty subset of the set of five
 vertices that form the chair, is a candidate for inclusion in the minimal
 modulator---where the number of undecided vertices decreases by \(5\) in each
 branch. The running time is thus
 \(\OhStar(31^{\frac{n}{5}})\approx{}1.987^{n}\), which is far worse than that
 of the simple algorithm for \probMR that we outlined above.

 Instead of looking for a ``good'' minimal modulator to chair-free graphs, we
 repeatedly find a ``fresh'' chair, with all undecided vertices. We then
 carefully branch on its vertices, assigning each vertex either to the
 ``independent set part'' or the ``cluster part'' of a putative monopolar
 partition. We stop the branching when \emph{no} induced chair has \emph{all}
 its vertices undecided; we then apply \autoref{thm:MEClawFreeDS} to solve the
 remaining instance. The careful branching and early stopping lead to the
 considerable speed-up from the trivial \(\OhStar(2^{n})\) to the
 \(\OhStar(1.3734^{n})\) of \autoref{thm:exactAlgo} and
 \autoref{thm:exactAlgoLMP}. We get our fast \classFPT algorithms by ``branching
 towards'' \autoref{thm:MEClawFreeDS}, as well.

\paragraph*{Related work.}
Though the notion of monopolar graphs arose from purely theoretical
considerations , they have been found to have some practical application as
well. The \ProblemFormat{Monopolar Editing} problem---adding or deleting the
smallest number of edges from/to a given graph to make it monopolar---has
applications in solving problems on protein interaction
networks~\cite{bruckner2015graph}. The \probMR problem, which is the main focus
of our work, has received a considerable amount of attention starting around the
year \(2008\). 

 \paragraph*{Polynomial-time algorithms and complexity.} Ekim et
 al.~\cite{ekim2008polar} showed in \(2008\) that \probMR can be solved in
 \(\Oh(n)\) time on \emph{cographs} with \(n\) vertices\footnote{They used a
   different, more general definition of monopolarity; but one of their
   results---see their Lemma~11 and the preceding definitions---applies to our
   definition of monopolar graphs.}. Ekim and a different set of
 authors~\cite{ekim2008polarity} showed in \(2008\) that both \probMR and \probME
 can be solved in \(\Oh(n+m)\) time on \emph{chordal graphs} with \(n\) vertices
 and \(m\) edges\footnote{The definition of monopolarity became the one that we
   use, from this point on.}. In \(2010\) Ekim and
 Huang~\cite{ekim2010recognizing} described an \(\Oh(n)\)-time algorithm that
 determines if the \emph{line graph of a given bipartite graph} \(G\) on \(n\)
 vertices is monopolar. Churchley and Huang~\cite{churchley2011line} generalized
 this to \emph{all line graphs} in \(2011\): They designed an \(\Oh(n)\)-time
 algorithm that decides if the line graph of a given---arbitrary---graph \(G\) on
 \(n\) vertices is monopolar. Together with the well-known \(\Oh(n+m)\)-time
 algorithm for identifying the root graph of a line graph on \(n\) vertices and
 \(m\) edges~\cite{lehot1974optimal}, this gives an \(\Oh(n+m)\)-time algorithm
 for solving \probMR in line graphs.

 In \(2011\) Le and Nevries~\cite{le2011recognizing,LE20141} derived a number of
 results for \probMR on \emph{planar graphs}, including the first
 \classNP-hardness results for the problem in restricted graph classes. They
 showed that \probMR is \NPHard (i) on triangle-free planar graphs of maximum
 degree \(3\) and, for each fixed integer \(k\geq{}4\), (ii) on planar graphs of
 maximum degree \(3\) that exclude \(\{C_{4},\dotsc,C_{k}\}\) as induced
 subgraphs. On the positive side, they showed that \probMR, and more generally,
 \probME, are polynomial-time solvable on \emph{\(P_{5}\)-free graphs}, on a
 certain superclass of \emph{chair-free graphs}, and on a certain superclass of
 \emph{hole-free graphs}. They showed also that \probMR can be solved in
 polynomial time on \emph{maximal planar graphs}. For designing their algorithms
 they introduced a reduction from \probME---in a certain class of graphs---to
 \probTwoSAT, which we also use in designing our algorithms.

 Churchley and
 Huang~\cite{churchley2012list,churchley2014polarity,churchley2014solving} proved
 a number of results on monopolarity starting \(2012\). They showed that \probMR
 is \NPHard in triangle-free graphs, and derived two distinct polynomial-time
 algorithms for \probMR on claw-free graphs. In fact, one of these algorithms
 solves the more general \probLMP---which is equivalent to \probME, as we noted
 above---on claw-free graphs in polynomial time. They showed also that \probMR
 can be solved in polynomial time in a certain class of graphs that contains all
 of \emph{claw-free graphs, chordal graphs, cographs, permutation graphs, and
   co-comparability graphs}. Independently of this work, Ekim et
 al.~\cite{ekim2013polar} showed in \(2013\) that \probMR can be solved in
 polynomial time on permutation graphs.

 \paragraph*{Parameterized algorithms.} We are aware of only two papers that
 take up the parameterized complexity of \probMR. These are both by the same set
 of authors, namely, Kanj, Komusiewicz, Sorge and
 van~Leeuwen~\cite{kanj2018parameterized,kanj2020solving}. In the first paper
 from \(2018\), they show---\emph{inter alia}---that \probMR is \classFPT
 parameterized by the \emph{number of cliques} on the ``cluster side'' \(G[C]\)
 of a(n unknown) monopolar partition \(V(G)=C\uplus{}I\) of the input graph
 \(G\). They derive an algorithm that, given a graph \(G\) and a positive
 integer \(k\), decides in \(\OhStar(2^{k})\) time whether \(G\) has a monopolar
 partition \(V(G)=C\uplus{}I\) where the ``cluster side'' \(G[C]\) is a disjoint
 union of at most \(k\) cliques. They show also that under ETH, any algorithm
 that solves \probMR must take \(2^{\Omega(n+m)}\) time on graphs with \(n\)
 vertices and \(m\) edges. In the second work, from \(2020\), they show that
 \probMR has a kernel with \(\Oh(k^{4})\) vertices, for the same parameter
 \(k\). Note that this parameter \(k\) is not comparable with either of our
 parameters \(k_{v},k_{e}\); there are easily seen to be graphs in which each of
 \(k_{v}\) or \(k_{e}\) is a small constant while \(k\) is unbounded, and
 \emph{vice versa}.

 \paragraph*{Organization of the rest of the paper.}
 In the next section we list our notation and terminology, and prove some
 preliminary results. We describe our exact exponential algorithms for
 monopolarity and prove \autoref{lem:GoodChairFreeDS}, \autoref{thm:exactAlgo},
 \autoref{thm:exactAlgoLMP} and \autoref{thm:MEClawFreeDS} in
 \autoref{sec:exact_chair-free}. We derive our \classFPT algorithms in
 \autoref{sec:fpt}, where we prove \autoref{thm:MRFPTClawFreeVD} and
 \autoref{thm:MRFPTClawFreeED}. We summarize our results and list some open
 problems in \autoref{sec:conclusion}.


\captionsetup[subfigure]{labelformat=simple}
\section{Preliminaries}\label{sec:prelims}
All our graphs are finite, simple, and undirected. We follow the graph notation
and terminology of Diestel~\cite{diestel2017graph}. Let \(G=(V,E)\) be a graph.
For a vertex subset \(X\) of \(G\) we use (i) \(G[X]\) for the subgraph of \(G\)
\emph{induced} by the set \(X\), (ii) \(G-X\) for the subgraph of \(G\) obtained
by deleting all the vertices---and their incident edges---in set \(X\), (iii)
\(N(X)\) for the set of all vertices \(v\notin{}X\) that are adjacent to some
vertex of \(X\) (the \emph{open neighbourhood} of \(X\)), and (iv) \(N[X]\) for
the \emph{closed neighbourhood} \((X\cup{}N(X))\) of \(X\). For an \emph{edge}
subset \(Y\) of \(G\) we use \(G-Y\) to denote the subgraph of \(G\) obtained by
deleting all the edges---and no vertices---in set \(Y\).

The \emph{degree} of a vertex \(v\) is \(deg(v)=|N(v)|\). Graph \(G\) is
\emph{\(H\)-free} for a graph \(H\) if there is no vertex subset \(X\) of \(G\)
such that \(G[X]\) is isomorphic to \(H\); that is, if \(H\) does not appear as
an \emph{induced subgraph} of \(G\). A graph \(G\) is a \emph{cluster graph} if
each connected component of \(G\) is a complete graph. It is well-known that
graph \(G\) is a cluster graph if and only if \(G\) is \(P_{3}\)-free (see
\autoref{subfig:p3}).

A partition \(V = (C \uplus{} I)\) of the vertex set of graph \(G\) is a
\emph{monopolar partition} of \(G\) if \(G[C]\) is a cluster graph and the set
\(I\) is an independent set in \(G\). Graph \(G\) is \emph{monopolar} if it has
a monopolar partition. For two (disjoint) subsets \(C',I'\) of the vertex set
\(V\) of graph \(G\), a monopolar partition \(V=C\uplus{}I\) of \(G\) with
\(C'\subseteq{}C\) and \(I'\subseteq{}I\) is a \emph{monopolar extension} of the
pair $(C',I')$. Graph $G$ is \emph{$(C',I')$-monopolar extendable} if it admits
a monopolar partition that is an extension of $(C',I')$.

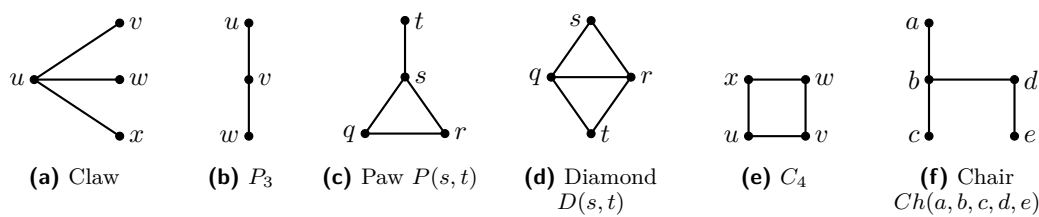
\begin{figure}
  \centering
      \subfloat[\centering Claw\label{subfig:claw}]{
        \centering
        \begin{tikzpicture}[scale=0.75]
          \filldraw (0,0) circle (2pt) node[left] {$u$};
          \filldraw (1.5,0) circle (2pt) node[right] {$w$};
          \filldraw (1.5,1) circle (2pt) node[right] {$v$};
          \filldraw (1.5,-1) circle (2pt) node[right] {$x$};
          \draw[thick] (0,0) -- (1.5,0);
          \draw[thick] (0,0) -- (1.5,1);
          \draw[thick] (0,0) -- (1.5,-1);
        \end{tikzpicture}
      }
      \hfill
      \subfloat[\centering $P_3$\label{subfig:p3}]{
        \centering
        \begin{tikzpicture}[scale=0.75]
          \filldraw (0,1) circle (2pt) node[left] {$u$};
          \filldraw (0,-1) circle (2pt) node[left] {$w$};
          \filldraw (0,0) circle (2pt) node[right] {$v$};

          \draw[thick] (0,0) -- (0,1);
          \draw[thick] (0,0) -- (0,-1);
        \end{tikzpicture}
      }
      \hfill
      \subfloat[\centering Paw \(P(s,t)\)\label{subfig:paw}]{
          \centering
          \begin{tikzpicture}[scale=0.75]
              \filldraw (0,0) circle (2pt) node[right] {$s$};
              \filldraw (-0.7,-1) circle (2pt) node[left] {$q$};
              \filldraw (0.7,-1) circle (2pt) node[right] {$r$};
              \filldraw (0,1) circle (2pt) node[right] {$t$};

              \draw[thick] (0,0) -- (0,1);
              \draw[thick] (0,0) -- (-0.7,-1);
              \draw[thick] (0,0) -- (0.7,-1);
              \draw[thick] (-0.7,-1) -- (0.7,-1);
          \end{tikzpicture}
      }
      \hfill
      \subfloat[\centering Diamond \(D(s,t)\)\label{subfig:diamond}]
      {\centering
        \begin{tikzpicture}[scale=0.75]
          \filldraw (0,1) circle (2pt) node[left] {$s$};
          \filldraw (0,-1) circle (2pt) node[right] {$t$};
          \filldraw (-0.7,0) circle (2pt) node[left] {$q$};
          \filldraw (0.7,0) circle (2pt) node[right] {$r$};

          \draw[thick] (0,1) -- (0.7,0);
          \draw[thick] (0,1) -- (-0.7,0);
          \draw[thick] (0,-1) -- (0.7,0);
          \draw[thick] (0,-1) -- (-0.7,0);
          \draw[thick] (0.7,0) -- (-0.7,0);
      \end{tikzpicture}
      }
      \hfill
      \subfloat[\centering $C_4$\label{subfig:c4}]{
        \centering
        \begin{tikzpicture}[scale=0.75]
          \filldraw (0,0) circle (2pt) node[left] {$u$};
          \filldraw (0,1) circle (2pt) node[left] {$x$};
          \filldraw (1,1) circle (2pt) node[right] {$w$};
          \filldraw (1,0) circle (2pt) node[right] {$v$};

          \draw[thick] (0,0) -- (0,1);
          \draw[thick] (0,0) -- (1,0);
          \draw[thick] (1,1) -- (1,0);
          \draw[thick] (1,1) -- (0,1);
      \end{tikzpicture}
      }
      \hfill
      \subfloat[\centering Chair \chair{a,b,c,d,e}\label{subfig:chair}]{
        \centering
        \begin{tikzpicture}[scale=0.75]
          \filldraw (0,1) circle (2pt) node[left] {$a$};
          \filldraw (0,0) circle (2pt) node[left] {$b$};
          \filldraw (0,-1) circle (2pt) node[left] {$c$};
          \filldraw (1.5,0) circle (2pt) node[right] {$d$};
          \filldraw (1.5,-1) circle (2pt) node[right] {$e$};
          \draw[thick] (0,0) -- (0,1);
          \draw[thick] (0,0) -- (0,-1);
          \draw[thick] (0,0) -- (1.5,0);
          \draw[thick] (1.5,0) -- (1.5,-1);
        \end{tikzpicture}
      }
      \caption{Some small named graphs.}\label{fig:smallgraphs}
\end{figure}

Graph $G$ is \emph{claw-free} if it does not have a
\emph{claw}---\autoref{subfig:claw}---as an induced subgraph. A subset \(X\) of
the vertex set of graph \(G\) is said to be a \emph{vertex modulator of \(G\) to
  claw-free graphs} if the graph \(G-X\) is claw-free. A subset \(Y\) of the
edge set of graph \(G\) is said to be an \emph{edge modulator of \(G\) to
  claw-free graphs} if the graph \(G-Y\) is claw-free. Graph $G$ is
\emph{chair-free} if it does not have a
\emph{chair}---\autoref{subfig:chair}---as an induced subgraph. A subset \(X\)
of the vertex set of graph \(G\) is said to be a \emph{vertex modulator of \(G\)
  to chair-free graphs} if the graph \(G-X\) is chair-free. \emph{We use the
  short forms ``vertex modulator'' and ``edge modulator'', when the fuller forms
  are clear from the context.}

We use \(P_{k}\) to denote a path with \(k\) vertices, and \(C_{k}\) for a cycle
with \(k\) vertices. Some other small graphs that we need to refer by name are
listed in \autoref{fig:smallgraphs}. We use \chair{a,b,c,d,e} to denote a \emph{chair}
with \(\deg(b)=3,\deg(d)=2,\deg(a)=\deg(c)=\deg(e)=1\) as in
\autoref{subfig:chair}, \(P(s, t)\) to denote a \emph{paw} with
\(\deg(s)=3,\deg(t)=1\), and $D(s, t)$ for a \emph{diamond} with \(\deg(s)=\deg(t)=2\).

Now, we show the following basic facts about monopolar extension:
\begin{lemma}\label{lem:deg1}
  Let \(C',I'\) be two (disjoint) subsets of the vertex set
  \(V\) of graph \(G\), and let \(x\notin(C'\cup{}I')\) be a vertex of degree
  \(1\) in \(G\). 
  Then G is $(C',I')$-monopolar extendable if and only if $G-\{x\}$ is
  $(C',I')$-monopolar extendable.
\end{lemma}
\begin{proof}
  Note that if $G$ is $(C',I')$-monopolar extendable then so is
  \(G-\{x\}\). 
  Suppose \(G-\{x\}\) is $(C',I')$-monopolar extendable with monopolar partition
  \(V-\{x\}=(C\uplus{}I)\). Let the lone neighbour of $x$ in \(G\) be $y$. If
  \(y\in{}C\) holds then we set \(\tilde{I}=I\cup{x}\); observe that
  \(V=C\uplus\tilde{I}\) is a valid monopolar partition of $G$. On the other
  hand if 
  \(y\in{}I\) holds then we set \(\tilde{C}=C\cup{x}\); observe that
  \(V=\tilde{C}\uplus{}I\) is then a valid monopolar patition of $G$.
\end{proof}

\begin{lemma}\label{lem:aside}
  A graph $G$ is \((C',I')\)-monopolar extendable if and only if $I'$ is an
  independent set in \(G\) and the induced subgraph \(G-I'\) is
  \(((C'\cup{}N(I')),\emptyset)\)-monopolar extendable.
\end{lemma}
\begin{proof}
  Suppose \(G=(V,E)\) is \((C',I')\)-monopolar extendable, and let
  \(V=(C\uplus{}I)\) be a monopolar partition of \(G\) where $C'\subseteq{}C$
  and $I'\subseteq{}I$ hold. Then---by definition--- $I'$ is an independent set
  in \(G\) and \(C\cap{}I'=\emptyset\). Since \(I\supseteq{}I'\) is an
  independent set, we get that every vertex in the set \(N(I')\) is in the set
  \((V\setminus{}I)=C\); that is, we get that \(N(I')\subseteq{}C\) holds. It
  follows that \((C'\cup{}N(I'))\subseteq{}C\) holds. And since
  \(\emptyset\subseteq{}(I\setminus{}I')\), we get that
  \(V=C\uplus{(I\setminus{}I')}\) is a monopolar extension of the pair
  \(((C'\cup{}N(I'),\emptyset))\) in the induced subgraph \(G-I'\). Thus the
  induced subgraph \(G-I'\) is \(((C'\cup{}N(I')),\emptyset)\)-monopolar
  extendable.

  Conversely, suppose $I'$ is an independent set in \(G\) such that the induced
  subgraph \(G-I'\) is \(((C'\cup{}N(I')),\emptyset)\)-monopolar extendable for
  some vertex subset \(C'\subseteq(V\setminus{}I')\). Let
  \((\tilde{C},\tilde{I})\) be a \(((C'\cup{}N(I')),\emptyset)\)-monopolar
  extension of \(G-I'\). Then both (i) \(C'\subseteq\tilde{C}\) and (ii)
  \(N(I')\subseteq{}\tilde{C}\) hold, by definition. Thus we get that
  \(N(I')\cap\tilde{I}=\emptyset\) holds, and hence that the set
  \((I'\cup\tilde{I})\) is an independent set in graph \(G\). It follows that
  \(V=(\tilde{C},(I'\cup\tilde{I}))\) is a monopolar partition of \(G\) that
  extends the pair \((C', I')\).
\end{proof}

\begin{lemma}\label{lem:lmp2me}
  {\normalfont{\probME}} and {\normalfont{\probLMP}} are computationally equivalent modulo polynomial time.
\end{lemma}
\begin{proof}
  Given an instance of \probLMP with list function $L$, we can construct an
  equivalent instance of \probME by defining $C'$ as the set of all vertices $v$
  such that $L(v) = \{{\hat{C}}\}$ and \(I'\) as the set of all vertices $u$
  such that $L(u)=\{\hat{I}\}$. Conversely, given an instance of \probME with
  sets $C'$ and $I'$, we can construct an equivalent instance of \probLMP by
  setting the list function $L$ to be $L(v)=\{\hat{C}\}$ if $v \in C'$,
  $L(v)=\{\hat{I}\}$ if $v \in I'$, and $L(v)=\{\hat{C},\hat{I}\}$ if $v \notin
  (C' \cup I')$.
\end{proof}

Le and Nevries~\cite{LE20141} prove the following result for claw-free graphs:
\begin{theorem}\label{thm:nevries2sat}
  {\normalfont{\probME}} can be solved in $O(n^4)$ time on $n$-vertex claw-free graphs.
\end{theorem}

As we noted in \autoref{sec:introduction}, the algorithms of
\autoref{thm:MRFPTClawFreeVD} and \autoref{thm:MRFPTClawFreeED} do not require
small-sized (vertex, or edge) modulators to be given as part of the input. But
they do require access to such sets, and so they compute the required modulators
as an intermediate step. That is, our algorithms that prove
\autoref{thm:MRFPTClawFreeVD} and \autoref{thm:FPTClawFreeED} employ algorithms
to solve the following parameterized deletion problems along the way:

\defparproblem{\probCFVD} {An undirected graph \(G=(V,E)\), and an integer
  \(k\).} {\(k\)} {If there is a set \(X\subseteq V\;;\;|X|\leq{}k\) such that
  deleting \(G - X\) is a claw-free graph, then output one such set \(X\).
  Otherwise, output \no. }

\defparproblem{\probCFED} {An undirected graph \(G=(V,E)\), and an integer
  \(k\).} {\(k\)} {If there is a set \(X\subseteq E\;;\;|X|\leq{}k\) such that
  deleting \(X\) from \(G\) results in a claw-free graph, then output one such
  set \(X\). Otherwise, output \no. }

It turns out that in each case, the main bottleneck for solving \probMR is in
fact the time required for computing one such modulator; a faster algorithm for
\probCFVD would directly imply a faster running time for
\autoref{thm:MRFPTClawFreeVD}, and similarly for \probCFED and
\autoref{thm:MRFPTClawFreeED}.

So we wanted to find fast \classFPT algorithms for \probCFVD and \probCFED. As
far as we could find, the fastest way to solve \probCFVD is by phrasing it as an instance of a
\ProblemName{$4$-Hitting Set} problem and then applying the best known \classFPT
algorithm for \ProblemName{$4$-Hitting Set} due to Dom et
al.~\cite[Theorem~3.1]{dom2010fixed}. This yields an algorithm that solves
\probCFVD in \(\OhStar(3.076^{k})\) time, and this running time appears in
\autoref{thm:MRFPTClawFreeVD}. A similar approach works for \probCFED; we
express it as an instance of \ProblemName{$3$-Hitting Set} and then apply the
argument underlying Dom et al.'s algorithm for this
problem~\cite[Proposition~3.1]{dom2010fixed} which uses an algorithm for \ProblemName{Vertex Cover} as black box. If we use the current best running time of \(\OhStar(1.253^{k})\) for \ProblemName{Vertex Cover}~\cite{harris_et_al:LIPIcs.STACS.2024.40} in this argument we get an algorithm that solves \probCFED in \(\OhStar(2.253^{k})\) time, and this running time appears in \autoref{thm:MRFPTClawFreeED}.

These two fast algorithms for modulators to claw-free graphs are thus direct
applications of existing results to these problems. We capture these as a
theorem for ease of reference and discovery.


\begin{reptheorem}{repthm:FPTClawFreeED}
  \label{thm:FPTClawFreeED}
  \normalfont{\probCFVD} can be solved in \(\OhStar(3.076^{k})\) time, and
  \normalfont{\probCFED} can be solved in \(\OhStar(2.253^{k})\) time.
\end{reptheorem}

\section{Exact Algorithms for Monopolarity}
\label{sec:exact_chair-free}

In this section we prove
\hyperlink{proof:exactAlgo}{Theorem~\ref{thm:exactAlgo}} and
\hyperlink{proof:exactAlgoLMP}{Theorem~\ref{thm:exactAlgoLMP}}.
We get these faster running times 
by exploiting various structural properties of monopolar graphs to lift a
theorem of Le and Nevries~\cite{LE20141} for chair-free graphs, so that it works
for all graphs; see
\hyperlink{proof:MEClawFreeDS}{Theorem~\ref{thm:MEClawFreeDS}}.
Le and Nevries~\cite{LE20141} showed that \probME, and consequently
\probMR, can be solved in polynomial time on chair-free graphs. 
They achieved this result by reducing \probME on (a superclass of) chair-free
graphs to \probTwoSAT; our algorithms make extensive use of this reduction.

The small graphs on four vertices depicted in \autoref{subfig:paw},
\autoref{subfig:diamond}, and \autoref{subfig:c4} on page~\pageref{subfig:c4},
are frequently considered in algorithms for monopolar recognition across various
graph classes~\cite{churchley2011line,LE20141}. This is because various useful
observations can be made about the vertices of these small graphs in any valid
monopolar partition of a graph \(G\), when the small graphs appear as induced
subgraphs of \(G\). Consider a graph \(G=(V,E)\) which admits a monopolar
partition \(V=C\uplus{}I\). We can easily infer the following properties
(Refer~\autoref{fig:smallgraphs} for the notation):

\begin{enumerate}\label{list:smallGraphProperties}
\item For any edge \(uv\in{}E\), either \(u\in{}C\) or \(v\in{}C\) must hold; if
  both \(u\) and \(v\) belong to \(I\) then it would violate the requirement
  that \(I\) must be an independent set.
\item For any induced paw \(P(s,t)\) in \(G\), at least one of the two vertices
  \(\{s,t\}\) must belong to the set \(I\). Indeed, suppose
  \(\{s,t\}\subseteq{}C\) holds. Then vertex \(q\) cannot be in set \(C\) since
  the edge \(qt\) is not present in \(G\). Similarly, vertex \(r\) cannot be in
  set \(C\) because of the missing edge \(rt\). Thus both the vertices \(q,r\)
  must be in set \(I\), but this violates the requirement that \(I\) must be an
  independent set.
\item For any induced diamond \(D(s,t)\) in \(G\), at least one of the two
  vertices \(\{s,t\}\) must belong to the set \(I\). Indeed, suppose
  \(\{s,t\}\subseteq{}C\) holds. Then neither of the vertices \(\{q,r\}\) can be
  in set \(C\) since the edge \(st\) is not present in \(G\). Thus both the
  vertices \(q,r\) must be in set \(I\), but this violates the requirement that
  \(I\) must be an independent set.
\item For a vertex $u \in C$ and any induced $P_{3}$ of $G$, say of form $uvw$, that contains $u$, at most one of
the other two vertices in that $P_{3}$ can belong to the set $C$, i.e., $|C \cap
\{v,w\}| \leq 1$; all three vertices of an induced \(P_{3}\) being in \(C\),
  would violate the requirement that \(G[C]\) be a cluster graph.
\item For any induced \(C_{4}\) of the form $uvwx$, one of the following must
  hold:
  \begin{itemize}
  \item Vertices \(u,w\) are in set \(I\), and vertices \(v,x\) are in set
    \(C\), \emph{OR}
  \item Vertices \(u,w\) are in set \(C\), and vertices \(v,x\) are in set
    \(I\).
  \end{itemize}
  That is, both the vertices in any one diagonal pair of vertices in the induced
  \(C_{4}\) belong to the same part in the monopolar partition, and the two
  pairs belong to different parts.

  Suppose not, and, say, \(u\in{}I\) and \(w\in{}C\) hold. Since \(u\in{I}\)
  holds, we get that both the vertices \(v,x\) must be in set \(C\). But then
  the vertices \(v,w,x\) induce a \(P_{3}\) in \(G[C]\), which violates the
  requirement that \(G[C]\) be a cluster graph. Note that because of the missing
  edges, all the four vertices \(u,v,w,x\) cannot belong to set \(C\); and
  because of the present edges, all four of them cannot belong to the set \(I\).
  These cover all the ways in which the property could potentially not hold.

\end{enumerate}

If the given graph \(G=(V,E)\) does not have an induced \(P_3\) 
then \(G\) is a cluster graph. In this case, for any vertex \(v\in{}V\), the
partition \(V=(V\setminus\{v\})\uplus\{v\}\) is a monopolar partition of \(G\),
and there is nothing more to do.
If the graph \emph{does} contain at least one induced \(P_{3}\), then we 
make use of the concept of a \emph{\(C'\)-good induced \(P_{3}\)}, a notion
introduced by Le and Nevries \cite{LE20141}. 

\begin{definition}\textup{\cite{LE20141}}\label{def:Cgood}
  Let \(G\) be a graph, and let \(C'\) be a subset of the vertex set of
  \(G\). We say that an induced $P_3$ in $G$ of the form $uvw$ is
  \emph{$C'$-good}, if it satisfies \emph{at least one} of the following conditions:
      \begin{enumerate}
      \item At least one of the vertices in the set
        \((N[v]\cup{}N[u]\cup{}N[w])\) is in the set \(C'\).
      \item At least one of the vertices \(u,v,w\) is a part of a triangle in $G$.
      \item At least one of the edges of the \(P_{3}\) $uvw$ is part of an
        induced $C_4$ in $G$.
      \end{enumerate}
      Otherwise, we say that the $P_{3}$ $uvw$ is $\emph{C'-bad}$. Finally, we say that graph $G$ is $\emph{C'-good}$ if every induced $P_{3}$ in $G$ is $\emph{C'-good}$.
\end{definition}

Le and Nevries~\cite{LE20141} show how, given an instance \((G,(C',\emptyset))\)
of \probME \emph{where \(G\) is \(C'\)-good}, we can construct in polynomial time a \probTwoSAT
formula \(F(G,C')\) where (i) the variables represent vertices of \(G\), (ii)
the clauses are constructed based on the properties of small induced graphs from \autoref{fig:smallgraphs}, and (iii) the formula \(F(G,C')\) is satisfiable if and only if
\((G,(C',\emptyset))\) is a \yes-instance of \probME. Note that this idea of
using only constant-sized induced subgraphs to completely characterize \probME
using a \probTwoSAT formula does not---unless \(\classP=\classNP\)---work for
graphs in general. It works for \(C'\)-good graphs because, as Le and Nevries
discovered, the concept of \(C'\)-goodness ensures that the graph's complexity
is limited in such a way that induced paws, diamonds, and \(C_{4}\)s provide
enough information to determine a valid monopolar partition.

Given an instance \((G,(\tilde{C},\tilde{I}))\) of \probME, first, we use
\autoref{lem:aside} to transform it into an equivalent instance
\((G',(C',\emptyset))\) where \(C' = \tilde{C} \cup N(\tilde{I})\) and $G' = G - \tilde{I}$. We then construct the corresponding \probTwoSAT formula \(F(G',C')\) as described below,
following Le and Nevries~\cite{LE20141}. Intuitively, the formula is set up in
such a way that the variable corresponding to a vertex \(v\) gets the value
\(1\) in a satisfying assignment of the formula if and only if \(v\) belongs to
the \emph{independent set} part in some valid monopolar extension.

\begin{enumerate}\label{formula:2sat}
    \item There is a Boolean variable \(v\) for each vertex \(v\) of \(G\).
    \item For each vertex \(u \in C'\), add \((\neg u)\) to the formula (a
      \emph{forced-cluster-clause}).
    \item For each edge \(uv\) in \(G\), add \((\neg u \lor \neg v)\) to the
      formula (an \emph{edge-clause}).
    \item For each induced paw \(P(s, t)\) and each induced diamond \(D(s, t)\)
      in \(G\), add \((s\lor{}t)\) to the formula (a \emph{paw-diamond-clause}).
    \item For each induced \(4\)-cycle \(C_4: u v w x\) in \(G\), add the
      clauses \((u\lor{}v)\), \((v\lor{}w)\), \((w\lor{}x)\), and \((x\lor{}u)\)
      to the formula (\emph{\(C_4\)-clauses}).
    \item For each vertex \(x\in{}C'\) and each induced \(P_{3}\) of the form
      \(xvw\) or \(vxw\), add the clause \((v \lor w)\) to the formula (a
      \emph{$P_3$-clause}).
\end{enumerate}
The formula \(F(G, C')\) is the conjunction of all forced-cluster-clauses,
edge-clauses, paw-diamond-clauses, \(C_4\)-clauses, and $P_3$-clauses. Le and
Nevries show that \emph{for a \(C'\)-good graph \(G\)}, the graph \(G\) is
monopolar extendable with respect to \((C',\emptyset)\) if and only if the
formula \(F(G, C')\) is satisfiable:

\begin{lemma}\textup{\cite[Lemma~1 and its proof]{LE20141}}\label{lem:2sat}
  Let $G = (V,E)$ be a graph and $C'$ be a subset of vertices such that $G$ is
  $C'$-good. Then in polynomial time we can construct an instance $F(G, C')$ of \probTwoSAT with the following properties. First, $G$ is $\left(C',\emptyset\right)$-monopolar extendable if and
  only if $F\left(G, C^{\prime}\right)$ is satisfiable. Second, if $F(G,C')$ is satisfiable then, the set $I$ of vertices corresponding to variables set to {\normalfont{True}} and the set $C$ of vertices corresponding to variables set to {\normalfont{False}} in a satisfying assignment to $F(G,C')$ together form a monopolar partition of $G$ i.e., $V = (C \uplus{} I)$ is a monopolar partition that extends $(C', \emptyset)$.
\end{lemma}

To solve \probMR on an input graph \(G\) using \autoref{lem:2sat} we need to
find a vertex subset \(C'\) of \(G\) such that graph \(G\) is \(C'\)-good. But
why should such a vertex set \(C'\) exist for \emph{every} monopolar graph
\(G\)? And even if it did, how could we find one such set? Le and
Nevries~\cite[Corollary~3]{LE20141} show that \emph{if the input graph \(G\) is
  chair-free} then they can use \autoref{lem:2sat} as a tool to solve an
instance \((G,(C',\emptyset))\) of \probME in polynomial time \emph{even if}
graph \(G\) is \emph{not} $C'$-good. We generalize this further; we show that in
fact we can solve an instance \((G,(C',\emptyset))\) of \probME in an
\emph{arbitrary graph} \(G\) in polynomial time if the set $C'$ is a
\emph{vertex modulator of \(G\) to chair-free graphs}, even if graph \(G\) is
\emph{not} $C'$-good. We thus reduce \probMR to the problem of finding such a
vertex modulator \(C'\) of \(G\), and then solving \probME on the instance
\((G,(C',\emptyset))\).

\hypertarget{proof:GoodChairFreeDS}{\repeatlemma{replem:GoodChairFreeDS}}
\addtocounter{lemma}{-1} 
\begin{proof}
  We first show that there is a vertex modulator that can play the role of
  \(C'\), and then we show that we may assume \(C'\) to be inclusion-minimal.

  If \(G\) is chair-free then the claim is trivially true, with
  \(C'=\emptyset\). So let \(G\) contain at least one induced chair, and let
  \(D\) be a vertex modulator of \(G\). Let \(V=(C\uplus{}I)\) be some monopolar
  partition of \(G\). If \(D\subseteq{}C\) then set \(C'=D\). If
  \(D\nsubseteq{}C\) then let \(Y=(I\cap{}D),\,Z=(C\cap{}D)\). Note that---since
  \(D=Y\uplus{}Z\) is a vertex modulator---each chair in the graph \(G-Z\)
  contains at least one vertex of \(Y\). It follows that the graph obtained from
  \(G-Z\) by deleting all the \emph{neighbours} of the vertices in \(Y\), is
  chair-free. That is, the set \(Z\cup{}N(Y)\) is a vertex modulator as well.
  Observe also that since \(Y\subseteq{}I\) holds, we have (i)
  \(N(Y)\cap{}Y=\emptyset\), and (ii) \(N(Y)\subseteq{}C\). Setting
  \(C'=(Z\cup{}N(Y))\) we get the required vertex modulator.

  Now let \(C'\subseteq{}C\) be a vertex modulator. As long as there is a vertex
  \(x\in{}C'\) such that \(C'\setminus\{x\}\) is a vertex modulator, set
  \(C'\gets{}C'\setminus\{x\}\). Once this procedure stops, we have a set \(C'\)
  that satisfies the conditions of the claim.
\end{proof}

\begin{lemma}\label{thm:chairhitting2sat}
  If a vertex subset \(C'\) is a vertex modulator of graph \(G=(V,E)\) to
  chair-free graphs, then the instance \((G,(C',\emptyset))\) of
  {\normalfont{\probME}} can be solved in polynomial time. Furthermore if graph
  \(G\) is \((C',\emptyset)\)-monopolar extendable, then a monopolar partition
  of \(G\) that extends $(C',\emptyset)$ can be obtained in polynomial time as
  well.
\end{lemma}
\begin{proof}
  Our proof follows the general approach that Le and Nevries used in the proof
  of their result on a superclass of chair-free
  graphs~\cite[Theorem~5]{LE20141}. Let \((G,(C',\emptyset))\) be the given
  instance of \probME, with \(G\) being a graph on \(n\) vertices and \(C'\)
  being a vertex modulator of \(G\). We preprocess the input using
  \autoref{lem:deg1} to ensure that every vertex of graph \(G\) which is not in
  the set \(C'\), has degree at least \(2\). Assume from now on that this
  condition holds for graph \(G\).

  We then check if graph \(G\) is \(C'\)-good. This takes \(\Oh(n^{4})\) time
  using a straightforward algorithm. If \(G\) is \(C'\)-good then we apply
  \autoref{lem:2sat} to construct the \probTwoSAT formula \(F(G,C')\) and check
  if the formula is satisfiable. We can construct the \probTwoSAT formula as per
  the description on page~\pageref{formula:2sat} in \(\Oh(n^{3})\) time, and
  check whether it is satisfiable in time which is linear in the size of the
  resulting formula~\cite{aspvall1982linear}. This part of the algorithm thus
  solves the problem in \(\Oh(n^{3})\) time, if graph \(G\) happens to be
  \(C'\)-good. If the formula is satisfiable, then we construct a monopolar
  partition of \(G\) that extends $(C',\emptyset)$ in polynomial time, using
  \autoref{lem:2sat}.

  So suppose graph \(G\) is not \(C'\)-good. Then the algorithm that checks for
  \(C'\)-goodness will find an induced \(P_{3}\), say \(uvw\), which is
  \emph{\(C'\)-bad}\footnote{See \autoref{def:Cgood}.}. We show, using a series
  of claims, that we can safely delete vertex \(v\) from \(G\) to get an
  equivalent instance of \probME.

  \begin{numclaim}\label{claim:vdegree2}
    The middle vertex \(v\) of the \(C'\)-bad \(P_{3}\) has degree exactly
    \(2\) in \(G\). 
  \end{numclaim}
  \begin{proof}
    Suppose not, and let \(x\notin\{u,w\}\) be a third neighbour of \(v\). Then
    vertex \(x\) cannot be in the set \(C'\), since that would contradict the
    first condition in the definition of a \(C'\)-bad \(P_{3}\). Vertex \(x\)
    cannot be adjacent to either of the two vertices \(\{u,w\}\), since this
    would contradict the second condition in the definition of a \(C'\)-bad
    \(P_{3}\). Consider the vertex $u$. As all vertices have degree at least
    \(2\), $u$ has a neighbour $y \neq v$ and $y \notin C'$. As $uvw$ is a
    \(C'\)-bad \(P_{3}\), $y$ cannot be equal to $w$ and it cannot be adjacent
    to $v$, $w$ or $x$. Thus we get that the vertices \(\{w,v,x,u,y\}\) form an
    induced chair \chair{w,v,u,x,y} that is disjoint from the set \(C'\); this
    contradicts the assumption that \(C'\) is a vertex modulator of graph \(G\).
  \end{proof}

  Consider the sets \(U=(N(u)\setminus\{v\}),W=N(w)\setminus\{v\} \).

  \begin{numclaim}\label{claim:disjointUWC}
    The sets \(U,W\), and \(C'\) are pairwise disjoint.
  \end{numclaim}
  \begin{proof}
    If there is a vertex \(x\) in \(G\) that belongs to both the sets \(U\) and
    \(W\), then either \(xuvw\) is an induced \(C_{4}\), or \(xuv\) is a
    triangle. The first implication contradicts the third condition in the
    definition  of a \(C'\)-bad \(P_{3}\). The second
    implication contradicts the second condition in the definition of a
    \(C'\)-bad \(P_{3}\). Thus we get that the sets \(U,W\) are disjoint.

    If there is a vertex \(x\) in \(G\) that belongs to both the sets \(U\)
    and \(C'\) then this contradicts the first condition in the definition of a
    \(C'\)-bad \(P_{3}\). Thus the sets \(U,C'\) are disjoint. A symmetric
    argument tells us that the sets \(W,C'\) are disjoint.
  \end{proof}

  \begin{numclaim}\label{claim:singletonsUW}
    The sets \(U,W\) contain exactly one vertex each: \(|U|=|W|=1\).
  \end{numclaim}
  \begin{proof}
    If the set \(U\) contains no vertex then vertex \(u\) has degree exactly
    \(1\) in graph \(G\), which contradicts the fact that every vertex of \(G\)
    which is not in the set \(C'\), has degree at least \(2\). Thus we get that
    \(|U|\geq{}1\) holds.

    Suppose the set \(U\) contains at least two vertices, and let \(x,y\) be two
    vertices in \(U\). Then \(xy\) is not an edge in \(G\), since otherwise the
    vertices \(u,x,y\) form a triangle in \(G\) and this contradicts the second
    condition for \(uvw\) being a \(C'\)-bad \(P_{3}\). Neither of the vertices
    \(x,y\) can be adjacent to vertex \(v\), since---given that \(uv\) is an
    edge---this would contradict the second condition in the definition of a
    \(C'\)-bad \(P_{3}\). These observations along with
    \autoref{claim:disjointUWC} tell us that the vertices \(x,u,y,v,w\) form an
    induced chair \chair{x,u,y,v,w} in \(G\) which is disjoint
    from the set \(C'\); this contradicts the assumption that \(C'\) is a
    vertex modulator of graph \(G\).

    Thus we get that \(|U|=1\) holds. A symmetric argument tells us that
    \(|W|=1\) holds.
  \end{proof}

  \begin{numclaim}\label{claim:deg2each}
    The two end-vertices \(u,w\) of the \(C'\)-bad \(P_{3}\), have degree exactly
    \(2\) each in \(G\).
  \end{numclaim}
  \begin{proof}
    Follows directly from \autoref{claim:singletonsUW} and the definitions of
    the sets \(U,W\).
  \end{proof}

  \begin{numclaim}\label{claim:reductionRule}
    Graph \(G\) is \((C',\emptyset)\)-monopolar extendable if and only if its
    subgraph \(G-\{v\}\) is \((C',\emptyset)\)-monopolar extendable.
  \end{numclaim}
  \begin{proof}
    Observe first that since \(v\notin{}C'\) holds, the set \(C'\) is a subset
    of the vertex set of graph \(G-\{v\}\). So it is valid to ask whether graph
    \(G-\{v\}\) is \((C',\emptyset)\)-monopolar extendable.

    Suppose graph \(G\) is \((C',\emptyset)\)-monopolar extendable, and let
    \(V=(C\uplus{}I)\) be a monopolar partition that extends \((C',\emptyset)\).
    It is straightforward to verify that
    \((V\setminus\{v\})=((C\setminus\{v\})\uplus{}(I\setminus\{v\}))\) is a
    monopolar partition of graph \(G-\{v\}\) that extends the pair
    \((C',\emptyset)\). Thus graph \(G-\{v\}\) is \((C',\emptyset)\)-monopolar
    extendable.

    Conversely, suppose graph \(G-\{v\}\) is \((C',\emptyset)\)-monopolar
    extendable, and let \((V\setminus\{v\})=(C\uplus{}I)\) be a monopolar
    partition of \(G-\{v\}\) that extends \((C',\emptyset)\). We show that we
    can add vertex \(v\) back to \(G-\{v\}\) in such a way as to get a monopolar
    partition \((\tilde{C}\uplus\tilde{I})\) of \(G\) that also extends
    \((C',\emptyset)\). From \autoref{claim:vdegree2} we know that vertex \(v\)
    has exactly two edges incident on it in graph \(G\), namely, \(vu\) and
    \(vw\). We consider the various cases of how vertices \(u,w\) are
    distributed among the sets \(C,I\) in the given monopolar partition of
    \(G-\{v\}\):
    \begin{description}
    \item[\(u\in{}C,w\in{}C\):] In this case we add vertex \(v\) to the
      independent set side; we set \(\tilde{C}=C,\tilde{I}=(I\cup\{v\})\). Set
      \(\tilde{I}\) is an independent set in \(G \) since both the neighbours of
      vertex \(v\) are in the set \(C\). Hence \((\tilde{C},\tilde{I})\) is a
      valid monopolar partition of graph \(G\).
    \item[\(u\in{}I,w\in{}I\):] In this case we add vertex \(v\) to the
      cluster side; we set \(\tilde{C}=(C\cup\{v\}),\tilde{I}=I\). Since both
      the neighbours of vertex \(v\) are not in the set \(C\), set \(\tilde{C}\)
      induces a cluster graph in \(G \) in which vertex \(v\) appears as a
      clique of size \(1\). Hence \((\tilde{C},\tilde{I})\) is a valid monopolar
      partition of graph \(G\).
    \item[\(u\in{}I,w\in{}C\):] Let \(x\) be the unique vertex (see
      \autoref{claim:singletonsUW}) in the set \(W\). We consider two sub-cases:
      \begin{description}
      \item[\(x\in{}I\):] In this case we add vertex \(v\) to the cluster side;
        we set \(\tilde{C}=(C\cup\{v\}),\tilde{I}=I\). The two vertices \(v,w\)
        form a clique of size \(2\) in \(G[\tilde{C}]\), since the only other
        neighbour \(u\) of \(v\) and the only other neighbour \(x\) of \(w\) are
        in the independent set side. Thus \((\tilde{C},\tilde{I})\) is a valid
        monopolar partition of graph \(G\).
      \item[\(x\in{}C\):] In this case we add vertex \(v\) to the cluster side
        \emph{and move} vertex \(w\) to the independent set side; we set
        \(\tilde{C}=((C\setminus{}\{w\})\cup\{v\}),\tilde{I}=(I\cup\{w\})\).
        Both the neighbours \(x,v\) of vertex \(w\) are in the set
        \(\tilde{C}\), so vertex \(w\) has no neighbours in the set
        \(\tilde{I}\). The only new neighbour \(v\) of vertex \(u\) is in the
        set \(\tilde{C}\), so vertex \(u\) has no neighbours in the set
        \(\tilde{I}\). It follows that \(\tilde{I}\) is an independent set in
        graph \(G\). Since both the neighbours \(u,w\) of the new vertex \(v\)
        are in the set \(\tilde{I}\), vertex \(v\) forms a clique of size one in
        the subgraph \(G[\tilde{C}]\). Thus \((\tilde{C},\tilde{I})\) is a valid
        monopolar partition of graph \(G\).
      \end{description}
    \item[\(u\in{}C,w\in{}I\):] This case is symmetric to the previous one.
    \end{description}
    Thus we can construct a monopolar partition of \(G\) that extends
    \((C',\emptyset)\) in every case, and this completes the proof of the claim.
  \end{proof}

  The above claims describe how to obtain a valid monopolar partition of $G$
  that extends $(C',\emptyset)$ from the monopolar partition of $G-\{v\}$ in
  polynomial time. Applying induction whose base case is where the graph is
  $C'$-good, we conclude that we can construct a monopolar partition of $G$ that
  extends $(C',\emptyset)$ in polynomial time if $C'$ is a vertex modulator of
  $G$, or conclude that no such partition exists.

  Our algorithm (procedure \textproc{USE-2-SAT}, \autoref{alg:use2sat} in
  Appendix~\ref{app:pseudocode}) for solving the instance \((G,(C',\emptyset))\)
  of \probME proceeds as follows: In \(\Oh(n^{4})\) time we either find that
  graph \(G\) is \(C'\)-good, or we find a \(P_{3}\) \(uvw\) in \(G\) which is
  not \(C'\)-good. In the former case we solve the instance as we described
  above. In the latter case we use \autoref{claim:reductionRule} to create an
  equivalent instance with one fewer vertex, and recurse on this reduced
  instance. This latter step can happen at most \(n\) times, so our algorithm
  solves the instance \((G,(C',\emptyset))\) of \probME where \(C'\) is a vertex
  modulator of \(G\) to chair-free graphs, in \(\Oh(n^{5})\) time where \(n\) is
  the number of vertices in \(G\). It is evident that we can construct the
  monopolar partition in polynomial time as well using the proofs of
  \autoref{claim:reductionRule} and \autoref{lem:2sat}.
\end{proof}

As \probLMP and \probME are polynomially equivalent as shown in
\autoref{lem:lmp2me}, we immediately get a similar result for \probLMP when the
set of vertices \(\{v : L(v)=\{\hat{C}\}\}\) forms a vertex modulator of the
input graph:

\begin{corollary}\label{cor:lmp2sat}
  Let \((G,L)\) be an instance of {\normalfont{\probLMP}} where the set
  \(C'=\{v : L(v)=\{\hat{C}\}\}\) of vertices that are assigned to the cluster
  side, forms a vertex modulator of graph \(G\) to chair-free graphs. Such an
  instance can be solved in polynomial time.
\end{corollary}

We now have all the ingredients for proving our generalization of Le and
Nevries' polynomial-time algorithm for \probME on chair-free graphs: 

\hypertarget{proof:MEClawFreeDS}{\repeattheorem{repthm:MEClawFreeDS}} 
\begin{proof}
  Let $(G,(C',I'))$ be an instance of \probME where $C'$ is a vertex modulator
  of $G$. Assuming this is an \yesinstance, $I'$ is an independent set. Hence, we can apply \autoref{lem:aside} to obtain an equivalent instance
  $(G',((C' \cup N(I')),\emptyset))$ of \probME where \(G'= G-I'\). As $C'$ is a
  vertex modulator of $G$, we get that \(C'\cup{}N(I')\) is a vertex modulator
  of $G'$ as well. Thus by \autoref{thm:chairhitting2sat}, the instance
  \((G',((C'\cup{}N(I')),\emptyset))\) of \probME can be solved in polynomial
  time. Furthermore if graph \(G\) is \(((C'\cup{}N(I')),\emptyset)\)-monopolar
  extendable, then a monopolar partition that extends
  \(((C'\cup{}N(I')),\emptyset)\) can be obtained in polynomial time as well.
  Thus the instance \((G,(C',I'))\) of \probME can be solved in polynomial time.
  In case $G$ is $(C',I')$ monopolar extendable, let
  \(V\setminus{}I'=(\tilde{C}\uplus{}\tilde{I})\) where
  \((C'\cup{}N(I'))\subseteq\tilde{C}\) be the monopolar partition returned by
  applying \autoref{thm:chairhitting2sat}. Now,
  \(V=(\tilde{C}\uplus{}(\tilde{I}\cup I'))\) is a monopolar partition of $G$
  that extends $(C',I')$. Thus a monopolar partition of $G$ that extends
  $(C',I')$ can be obtained in polynomial time as well.
\end{proof}

A straightforward approach to solving \probMR on arbitrary graphs \(G\) would be
to first find a vertex modulator. Next, we determine which subset of vertices
from this modulator should be assigned to \(C'\), while assigning the rest to
\(I'\). However, as noted in~\autoref{lem:aside}, it is sufficient to consider
whether \( G' = (G - I') \) is \(((C'\cup{}N(I')),\emptyset)\)-monopolar
extendable. This approach resolves the monopolarity problem for general graphs
on \(n\) vertices in \(\OhStar(2^n)\) time. Remarkably, it turns out that we can
solve this problem in \(\OhStar({1.3734}^{n})\) time by leveraging the
properties of chairs.

We show how to solve the more general \probME in \(\OhStar({1.3734}^{n})\) time.
Let \((G=(V,E),C',I')\) be an instance of \probME. If either the graph \(G\) is
chair-free or the set \(C'\) is a vertex modulator, then we can solve this
instance in polynomial time. Such instances form the base cases of \probME that
are easy to solve. We start with a simple observation:

\begin{lemma}\label{lem:ME_reduction}
  Let \((G=(V,E),C',I')\) be an instance of \probME. In polynomial time we can
  either solve the instance or derive an \emph{equivalent} instance
  \((\tilde{G},\tilde{C},\tilde{I})\) of \probME with the following property:
  there exists at least one induced chair---say \chair{a,b,c,d,e}---in
  \(\tilde{G}\) such that \(\{a,b,c,d,e\}\cap(\tilde{C}\cup\tilde{I})=\emptyset\).
\end{lemma}
\begin{proof}
  If graph \(G\) is chair-free then we apply
  \autoref{thm:chairhitting2sat}---with \(C'\) set to \(\emptyset\)---to solve
  the instance in polynomial time. So suppose that \(G\) contains one or more
  chairs as induced subgraphs. Set
  \(\tilde{G}\gets{}G,\tilde{C}\gets{}C',\tilde{I}\gets{}I'\). Now,
  \begin{enumerate}
  \item If \(v\in\tilde{I}\) holds for a vertex \(v\) then we---safely---add all
    neighbours \(x\) of \(v\) to the set \(\tilde{C}\).
  \item Suppose the chair in \autoref{subfig:chair} occurs as an induced
    subgraph of graph \(\tilde{G}\). Note that this subgraph is connected and
    has missing edges (e.g., \(ad\)). So there is no valid monopolar partition
    of \(\tilde{G}\) in which all the vertices of this induced chair belong to
    the set \(\tilde{C}\). Thus if \(\tilde{G}\) has an induced chair all of
    whose vertices belong to the set \(\tilde{C}\), then we---safely---return
    \no.
  \item If the set \(\tilde{C}\) is a vertex modulator, then we solve the
    instance in polynomial time using \autoref{thm:chairhitting2sat}.
  \end{enumerate}

  The resulting instance \((\tilde{G},\tilde{C},\tilde{I})\) is equivalent to
  the original instance \((G,C',I')\). Let \chair{a,b,c,d,e} be an induced chair
  in \(\tilde{G}\); such a chair must exist, or we would have solved this
  instance already. From point (3) we get that
  \(\{a,b,c,d,e\}\cap\tilde{C}=\emptyset\) holds. Point (1) then implies that
  \(\{a,b,c,d,e\}\cap\tilde{I}=\emptyset\) holds as well. Thus we get that
  \(\{a,b,c,d,e\}\cap(\tilde{C}\cup\tilde{I})=\emptyset\) holds for the chair
  \chair{a,b,c,d,e}.
\end{proof}

We now describe how to branch on such a chair, to eventually obtain an instance
where the set \(\tilde{C}\) is a vertex modulator.

So let \chair{a,b,c,d,e} be an induced chair in $G = (V,E)$ such that
\(\{a,b,c,d,e\}\cap(\tilde{C}\cup\tilde{I})=\emptyset\) holds. Let
\(V=(C\uplus{}I)\) be a(n unknown) valid monopolar partition of $G$. Our
algorithm for \probME does an exhaustive branching on the vertices $b$ and
$e$---see \autoref{subfig:chair} for the notation---being in $C$ or $I$:

\begin{description}
\item[$b \in{} C, e \in{} C$:] Since $G[C]$ is a cluster graph and the edge $be$
  is not present in graph \(G\), vertex $d$ cannot belong to the set $C$. So we
  have \(d\in{}I\). We cannot conclude anything about vertices $a$ and \(c\)
  except that they cannot be together in set \(C\).

\item[$b \in I, e \in C$:] Since $G[I]$ is edge-less and the edges $ab$, $bc$
  and $bd$ are present in graph \(G\), vertices $a$, $c$ and $d$ cannot belong
  to the set $I$. So we have $a \in C$, $c \in C$ and $d \in C$.

\item[$b \in C, e \in I$:] Vertex $d$ cannot belong to the set $I$ as $de$ is an
  edge and $e \in I$. Thus $d \in C$. Also, vertex $a$ cannot belong to the set
  $C$ as the edge $ad$ is not present and vertex \(c\) cannot belong to the set
  $C$ as the edge \(cd\) is not present. So we have $d \in C$, $a \in I$ and $c
  \in I$.

\item[$b \in I, e \in I$:] Here, $a$, $c$ and $d$ cannot belong to the set $I$
  as $ab$, $bc$ and $bd$ are edges and $b \in I$ where $G[I]$ is needed to be an
  independent set. So we have $a \in C$, $c \in C$ and $d \in C$.
\end{description}

Procedure \textproc{ME-Algorithm} (\autoref{alg:me} in
Appendix~\ref{app:pseudocode}) adopts the above branching strategy to solve
\probME. This procedure returns a monopolar partition if the input is a
\yes-instance.

Procedure \textproc{USE-2-SAT} (\autoref{alg:use2sat},
Appendix~\ref{app:pseudocode}) is a direct implementation of the proof of
\autoref{thm:chairhitting2sat}. So this procedure is correct, and it returns a
valid monopolar partition when the input is a \yes-instance. The correctness of
procedure \textproc{ME-Algorithm} (\autoref{alg:me}) now follows from the fact
that it does exhaustive branching on valid partitions, and from
\autoref{cor:lmp2sat}, \autoref{thm:nevries2sat}, and
\autoref{thm:chairhitting2sat}.

\begin{lemma}\label{lem:exactruntime}
  The call {\normalfont\(\textproc{ME-Algorithm}(G,C',I')\)} terminates in
  $\OhStar(1.3734^{n})$ time where \(n\) is the number of vertices in graph
  \(G\).
\end{lemma}
\begin{proof}
  Let \(\mu=\mu(G,C',I')=|V(G)\setminus(C'\cup{}I')|\) be a measure function
  that denotes the number of vertices in graph \(G\) which are not yet in the
  sets \(C',I'\). Note that the condition \(\mu\leq{}n\) holds before the
  procedure \textproc{ME-Algorithm} (see \autoref{alg:me}) makes any recursive
  call. The procedure makes recursive calls only when it finds a chair
  \(\mathcal{C}\), all of whose five vertices contribute to the measure \(\mu\).
  The three recursive calls in lines~\(17\), lines~\(20\) and~\(23\) cause
  \(\mu\) to drop by \(5\) each, since all five vertices of \(\mathcal{C}\) are
  moved to \(C'\cup{}I'\) in these calls. The recursive call in line~\(14\)
  results in \(\mu\) dropping by \(3\), since only three
  vertices---\(b,e,d\)---are moved to \(C'\cup{}I'\) by this call.

  The procedure stops recursing only when for each chair in graph \(G\), at
  least one vertex in the chair is in \(C'\cup{}I'\). If a vertex \(x\) in a
  chair \(\mathcal{C}\) is moved to the set \(I'\), then the procedure moves
  every neighbour of \(x\) to the set \(C'\). It follows that when the recursion
  stops, the set \(C'\) forms a vertex modulator of graph \(G\). The procedure
  solves this instance in polynomial time using \textproc{USE-2-SAT} from
  \autoref{alg:use2sat}, which implements \autoref{cor:lmp2sat}.

  Let \(T(\mu)\) denote the number of nodes in the recursion tree of the
  procedure, when invoked with an instance with measure \(\mu\). The recurrence
  for \(T(\mu)\) is then \(T(\mu) = 3T(\mu-5) + T(\mu-3)\). This solves to
  \(T(\mu)\leq{}{({1.3734})}^{\mu}\), and the lemma follows.
\end{proof}

We can convert an instance \(G=(V,E)\) of \probMR to an instance
\((G=(V,E),C',I')\) of \probME by setting \(C'=\emptyset,I'=\emptyset\). Thus we
get:

\hypertarget{proof:exactAlgo}{\repeattheorem{repthm:exactAlgo}}

Combining \autoref{lem:lmp2me} and \autoref{lem:exactruntime} we get:

\hypertarget{proof:exactAlgoLMP}{\repeattheorem{repthm:exactAlgoLMP}}

\section{FPT Algorithms for Monopolarity}
\label{sec:fpt}

In this section we prove
\hyperlink{proof:MRFPTClawFreeVD}{Theorem~\ref{thm:MRFPTClawFreeVD}} and
\hyperlink{proof:MRFPTClawFreeED}{Theorem~\ref{thm:MRFPTClawFreeED}}. We have
seen in \autoref{thm:FPTClawFreeED} that we can find a vertex modulator of graph
\(G\) to claw-free graphs in time $\OhStar({3.076}^{k_v})$ where $k_v$ is the
size of a smallest such modulator of \(G\). It also says that we can find an
\emph{edge} modulator of \(G\) to claw-free graphs in time
$\OhStar({2.253}^{k_e})$ where $k_e$ is the size of a smallest such modulator.
Once we have such sets, we show that we can solve monopolarity in arbitrary
graphs by making use of \autoref{thm:chairhitting2sat}\footnote{Observe that a
  \emph{claw-free} graph is also a \emph{chair-free} graph.}. The idea is to
find such sets and then guess the partition the elements of these sets are
assigned to in any valid monopolar partition and then try to solve monopolar
extension to decide monopolarity.

\hypertarget{proof:MRFPTClawFreeVD}{\repeattheorem{repthm:MRFPTClawFreeVD}}
\begin{proof}
  As proved in \autoref{thm:FPTClawFreeED}, we can solve \probCFVD in time
  $\OhStar({3.076}^{k_v})$, where ${k_v}$ is the smallest number of vertices
  that need to be deleted from $G = (V,E)$ to obtain a \emph{claw-free} graph.
  Let $C'$ be the vertex modulator obtained from the algorithm in
  \autoref{thm:FPTClawFreeED}. The vertices in $C'$ can be either part of $C$ or
  $I$ in a valid monopolar partition of the graph $G$, if it exists. So, we
  branch on the two possibilities -- a vertex $v \in C'$ is assigned to $C$ or
  $v$ is assigned to $I$. Thus we have $2^{k_v}$ different assignments of these
  vertices. We consider each of these assignments as separate instances. In each
  instance, let $\tilde{C}$ and $\tilde{I}$ be the set of vertices in \(C'\)
  that are assigned to $C$ and $I$ respectively. To see if $G$ admits a
  monopolar partition that extends $(\tilde{C},\tilde{I})$, we consider the
  equivalent problem of checking if $G - \tilde{I}$ is $((\tilde{C}\cup
  N(\tilde{I})),\emptyset)$-monopolar extendable by applying~\autoref{lem:aside}
  and $\tilde{I}$ is an independent set.
    \begin{claim}
      If $\tilde{I}$ is an independent set, then the set $\bar{C} = (\tilde{C} \cup{} N(\tilde{I}))$ is a vertex
      modulator to claw-free graphs of the graph $G-\tilde{I}$.\label{claim:clawfreevd}
    \end{claim}
    \begin{proof}
      We know that $C' = \tilde{C} \cup{} \tilde{I}$ is a vertex modulator. If
      every induced claw in $G$ contains a vertex $v \in \tilde{C}$, then it is
      obvious that $C'$ is a vertex modulator. If not, then there is at least
      one induced claw---$uvwx$---that contains vertices only in $\tilde{I}$.
      For simplicity let $\mathcal{C} = uvwx$ be an induced claw where $v \in
      \tilde{I}$ and $u,w,x \notin (\tilde{C} \cup{} \tilde{I})$. Now, $N(v) =
      {u} \implies u \in N(\tilde{I})$. Hence $\bar{C}$ contains at least one
      vertex from every induced claw in $G$. This proves the claim.
    \end{proof}
    So, we check if the graph $G$ with a particular assignment is $((\tilde{C}
    \cup{} N(\tilde{I})),\emptyset)$-monopolar extendable in polynomial time, by
    using the algorithm in \autoref{thm:chairhitting2sat}. The graph $G$ admits
    a monopolar partition iff there is atleast one assignment among the
    $2^{k_v}$ assignments for the vertex modulator $C'$ that can be extended to
    a monopolar partition. Thus it takes $\OhStar(2^{k_v})$ time to check if $G$
    is a monopolar graph given a vertex modulator $C'$. Combined with the
    initial algorithm to find the vertex modulator in time
    $\OhStar({3.076}^{k_v})$, the total running time of the algorithm is
    $\OhStar({3.076}^{k_v}+2^{k_v}) = \OhStar({3.076}^{k_v})$ which proves the
    theorem.
\end{proof}

\hypertarget{proof:MRFPTClawFreeED}{\repeattheorem{repthm:MRFPTClawFreeED}}
\begin{proof}
  As shown in \autoref{thm:FPTClawFreeED}, we can solve \probCFED in time
  $\OhStar({2.253}^{k_e})$, where ${k_e}$ is the smallest number of edges that
  need to be deleted from $G$ to obtain a \emph{claw-free} graph. Let $E'$ be
  the edge modulator to claw-free graphs, obtained from the algorithm in
  \autoref{thm:FPTClawFreeED}, we construct a set $C'$ by arbitrarily picking
  one vertex from each of these edges in the set $E'$.
    \begin{claim}\label{lem:cfedtocfed}
      Let $E$ be an edge modulator to claw-free graphs. Any set $D$ which
      contains at least one endpoint $u$ or $v$ of every edge $uv\in{}E$ is a
      \emph{vertex} modulator to claw-free.
    \end{claim}
    \begin{proof}
      Suppose not, then there exists an induced claw---$uvwx$---in $G-D$, and none
      of the edges \(\{uv,uw,ux\}\) is in the set \(E\). Hence, $G-E$ still
      contains the induced claw $uvwx$, which contradicts the assumption that $E$
      is an edge modulator. Thus, $D$ must be a vertex modulator.
    \end{proof}
    Thus the set $C'$ is a vertex modulator and $|C'| \leq k_e$. Now we proceed
    similarly to the proof of \autoref{thm:MRFPTClawFreeVD} to conclude that we
    can decide if $G$ is monopolar given an additional input $C'$ in time
    $2^{k_e}$.

    Combined with the initial algorithm to find the edge modulator in time
    $\OhStar({2.253}^{k_e})$, the total running time of the algorithm is
    $\OhStar({2.253}^{k_e}+2^{k_e}) = \OhStar({2.253}^{k_e})$ which proves the
    theorem.
\end{proof}


\section{Conclusion}
\label{sec:conclusion}

We derive fast exponential-time and \classFPT algorithms for \probMR. Our exact
algorithm for \probMR solves the problem in $\OhStar(1.3734^{n})$ time on graphs
with $n$ vertices, significantly improving the trivial $\OhStar(2^n)$-time
algorithm and the more involved \(\OhStar(3^{\frac{n}{3}})\)-time algorithm. We
also show how to solve the more general problems \probLMP and \probME in
$\OhStar(1.3734^{n})$ time. These are the fastest known exact algorithms for
these three problems. We derive two \classFPT algorithms for \probMR with two
notions of distance from triviality as the parameters. We show that \probMR can
be solved in \(\OhStar(3.076^{k_{v}})\) and \(\OhStar(2.253^{k_{e}})\) time,
where \(k_{v}\) and \(k_{e}\) are, respectively, the sizes of the smallest
vertex and edge modulators of the input graph to claw-free graphs. These results
are a significant addition to the small number of \classFPT algorithms known for
\probMR.

Le and Nevries~\cite{LE20141} showed that if a graph \(G\) is chair-free then
the instance \((G, (C',\emptyset))\) of \probME can be solved in polynomial time
for any vertex subset $C'$ of graph \(G\). We significantly generalize this
result: we show that the instance \((G, (C',\emptyset))\) of \probME can be
solved in polynomial time for \emph{any} graph $G$ if $C'$ is a vertex modulator
of \(G\) to chair-free graphs. This generalization of Le and Nevries' result
forms the basis of our fast algorithms.

\paragraph*{Open problems.} Le and Nevries~\cite{LE20141} show that \probMR can
be solved in polynomial time for a strict superclass of claw-free graphs, namely
the class of \(\{F_{1},F_{2},F_{3}\}\)-free graphs for certain small graphs
\(F_{1},F_{2},F_{3}\). This class includes chair-free graphs. A natural question
to ask is whether we can obtain a faster exact algorithm by proving the
equivalent of \autoref{thm:chairhitting2sat} for \(\{F_{1},F_{2},F_{3}\}\)-free
graphs, and then branching on induced subgraphs \(F_{1},F_{2},F_{3}\). Recall
that algorithms that solve these problem in time which is sub-exponential in
\(n\), are ruled out by ETH. It would therefore be interesting to establish a
lower bound for the base of the exponential term, perhaps based on SETH.

We can ask analogous questions in the parameterized setting as well: whether we
can use the sizes of vertex or edge modulators to $\{F_1, F_2, F_3\}$-free
graphs as the parameters, to get faster \classFPT algorithms. This latter
question is made interesting by the fact that each of the graphs
\(\{F_{1},F_{2},F_{3}\}\) has more vertices and edges than a claw. Using the
known \ProblemName{$d$-Hitting Set} \classFPT algorithms as a black box to solve
\probMR, similar to our algorithms for \autoref{thm:MRFPTClawFreeVD} and
\autoref{thm:MRFPTClawFreeED}, would therefore incur a significantly higher
running time in terms of the respective hitting set sizes. We will need clever
ideas to get around this slowdown.

\section*{Acknowledgements}
We extend our heartfelt gratitude to the anonymous reviewers of an earlier
version of this manuscript. Their detailed and insightful feedback has helped us
greatly improve the presentation.

\newpage
\bibliography{monopolarity}
\newpage
\appendix
\section{Pseudocode}
\label{app:pseudocode}
\begin{algorithm}[!ht]
  \small
  \caption{Algorithm that implements \autoref{thm:chairhitting2sat}}\label{alg:use2sat}
  \begin{algorithmic}[1]
    \Require{A graph $G$ along with a subset of vertices $C'$ such that
      $C'$ is a vertex modulator of \(G\) to chair-free graphs.}
    \Ensure{(\yes, $(C,I)$) if $G$ admits a monopolar partition
      $V = (C \uplus{} I)$, that extends $(C',\emptyset)$, else \no.}
    \Procedure{USE-2-SAT}{$G,C'$}
      \If{$G$ is $C'$-good}
        \State{$F(G,C') \leftarrow$ \probTwoSAT formula as described in \autoref{thm:chairhitting2sat}}
        \If{$F(G,C')$ is satisfiable}\Comment{Use any standard \probTwoSAT algorithm}
          \State{$t \leftarrow \text{ a satisfying assignment to variables (vertices of $G$) of } F(G,C')$}
          \State{$C \leftarrow \{v : t(v) = \text{ False}\}$}
          \State{$I \leftarrow \{v : t(v) = \text{ True}\}$}
          \State{\textbf{return} (\yes, $(C,I)$)}
        \EndIf{}
        \State{\textbf{return} \no}
      \EndIf{}
      \State{$uvw \leftarrow$ a $C'$-bad induced $P_3$ in $G$} \Comment{This is the case where $G$ is $C'$-bad}
      \State{$G' \leftarrow G - \{v\}$} \Comment{Based on \autoref{thm:chairhitting2sat}}
      \State{$X \leftarrow$ \Call{USE-2-SAT}{$G',C'$}}
      \If{$X = \no$}
        \State{\textbf{return} \no}
      \ElsIf{$X = (\yes, (C,I))$}
        \If{$u \in C \land w \in C$} \Comment{Below are the cases discussed in \autoref{claim:reductionRule}}
          \State{$\tilde{C} = C, \tilde{I} = (I \cup \{v\})$}
        \ElsIf{$u \in I \land w \in I$}
          \State{$\tilde{C} = C \cup \{v\}, \tilde{I} = I$}
        \ElsIf{$u \in I \land w \in C$}
          \State{$x \leftarrow $ unique neighbour of $w$ other than $v$}
          \If{$x \in I$}
            \State{$\tilde{C} = C \cup \{v\}, \tilde{I} = I$}
          \Else
            \State{$\tilde{C} = (C \setminus \{w\}) \cup \{v\}, \tilde{I} = I \cup \{w\}$}
          \EndIf{}
        \ElsIf{$u \in C \land w \in I$}
          \State{$x \leftarrow $ unique neighbour of $u$ other than $v$}
          \If{$x \in I$}
            \State{$\tilde{C} = C \cup \{v\}, \tilde{I} = I$}
          \Else
            \State{$\tilde{C} = (C \setminus \{u\}) \cup \{v\}, \tilde{I} = I \cup \{u\}$}
          \EndIf{}
        \EndIf{}
        \State{\textbf{return} (\yes, $(\tilde{C},\tilde{I}))$}
      \EndIf{}
    \EndProcedure{}
  \end{algorithmic}
\end{algorithm}

\begin{algorithm}
  \caption{Exact Exponential-time Algorithm for \probME}\label{alg:me}
  \begin{algorithmic}[1]
    \Require{Graph $G = (V,E)$ and disjoint vertex subsets \(C'\subseteq{}V,I'\subseteq{}V\).}
    \Ensure{(\yes, $(C,I)$) if $G$ admits a list monopolar partition $V = (C \uplus{} I)$ that extends \((C',I')\), else \no.}
    \Procedure{ME-Algorithm}{$G,C',I'$}
      \If{$G$ is a \textit{chair-free} graph \textbf{or} \((C'\cup{}N(I'))\) is
        a vertex modulator of \(G\)}
        \State{$C'\gets{}C'\cup{}N(I')$, $G'\gets{}G-I'$}\label{line:chair_free_block_start} \Comment{Implementation of \autoref{lem:aside}}
        \State \(U\gets\)\Call{USE-2-SAT}{$G',C'$}
        \If{\(U\neq\no\)}
          \State{\((\yes,(C,I))\gets{}U\), \(I \gets I\cup\tilde{I}\)}
          \State{\textbf{return} \((\yes,(C,I))\)}
        \Else{} \textbf{return} \no
        \EndIf{}
      \EndIf{}\label{line:chair_free_block_end}
      \State{Find an induced chair \chair{a,b,c,d,e}\;;\; 
        \(\{a,b,c,d,e\}\cap(C'\cup{}I')=\emptyset\).}
      \State{\(C_{1}\gets{}C'\),\(I_{1}\gets{}I'\)}
      \State{\(C_{1}\gets{}C_{1}\cup\{b,e\}\), \(I_{1}\gets{}I_{1}\cup\{d\}\)} \Comment{Case $b\in C$, $e \in C$}
      \State{\(A_{1}\gets\)\Call{ME-Algorithm}{$G,C_{1},I_{1}$}}
      \State{\(C_{2}\gets{}C'\),\(I_{2}\gets{}I'\)}
      \State{\(C_{2}\gets{}C_{2}\cup\{a,c,d,e\}\), \(I_{2}\gets{}I_{2}\cup\{b\}\)} \Comment{Case $b\in I$, $e \in C$}
      \State{\(A_{2}\gets\)\Call{ME-Algorithm}{$G,C_{2},I_{2}$}}
      \State{\(C_{3}\gets{}C'\),\(I_{3}\gets{}I'\)}
      \State{\(C_{3}\gets{}C_{3}\cup\{b,d\}\), \(I_{3}\gets{}I_{3}\cup\{a,c,e\}\)} \Comment{Case $b \in C$, $e \in I$}
      \State{\(A_{3}\gets\)\Call{ME-Algorithm}{$G,C_{3},I_{3}$}}
      \State{\(C_{4}\gets{}C'\),\(I_{4}\gets{}I'\)}
      \State{\(C_{4}\gets{}C_{4}\cup\{a,c,d\}\), \(I_{4}\gets{}I_{4}\cup\{b,e\}\)} \Comment{Case $b \in I$, $e \in I$}
      \State{\(A_{4}\gets\)\Call{ME-Algorithm}{$G,C_{4},I_{4}$}}
      \If{$A_1 \neq \no$} \textbf{return} \(A_{1}\)
      \ElsIf{$A_2 \neq \no$} \textbf{return} \(A_{2}\)
      \ElsIf{$A_3 \neq \no$} \textbf{return} \(A_{3}\)
      \ElsIf{$A_4 \neq \no$} \textbf{return} \(A_{4}\)
      \Else{} \textbf{return} \no
      \EndIf{}
    \EndProcedure{}
  \end{algorithmic}
\end{algorithm}


\end{document}